\documentclass[11pt]{article}

\usepackage{url}
\usepackage{epsf}
\usepackage{graphicx}
\usepackage{amsfonts}
\usepackage{amssymb}
\usepackage{amsmath}
\usepackage{latexsym}
\usepackage{color}
\usepackage{multirow}
\usepackage{setspace, float}
\usepackage{refcount}
\usepackage{thmtools,hyperref}
\usepackage{makeidx}  

\usepackage{thmtools,thm-restate}

\newcommand{\cone}{\mathcal{C}}
\newcommand{\T}{\mathcal{T}}
\newcommand{\R}{\mathcal{R}}
\newcommand{\spp}{\xi}
\newcommand{\pp}{p}

\newtheorem{theorem}{{\bf Theorem}}

\newtheorem{lemma}[theorem]{Lemma}

\newtheorem{property}[theorem]{Property}

\newcommand{\arr}[1]{\ensuremath{\protect\overrightarrow{#1}}}

\newenvironment{proof}[1][Proof]{\begin{trivlist}
\item[\hskip \labelsep {\bfseries #1}]}{\end{trivlist}}

\newcommand{\qed}{\nobreak \ifvmode \relax \else
      \ifdim\lastskip<1.5em \hskip-\lastskip
      \hskip1.5em plus0em minus0.5em \fi \nobreak
      \vrule height0.5em width0.5em depth0.25em\fi}

\begin{document}

\title{Spanning Properties of Theta-Theta-6}

\author{Mirela Damian\thanks{Department of Computing Sciences, Villanova University,
USA. \texttt{mirela.damian@villanova.edu}} 
\and 
John Iacono 
\thanks{Universit\'e Libre de Bruxelles and New York University. 
Supported by NSF grants CCF-1319648, CCF-1533564, CCF-0430849 and MRI-1229185, a Fulbright Fellowship and by the Fonds de la Recherche Scientifique-FNRS under Grant no MISU F 6001 1.}
\and 
Andrew Winslow\thanks{University of Texas Rio Grande Valley, USA. \url{andrew.winslow@utrgv.edu}}}
\date{}
\maketitle

\begin{abstract}
We show that, unlike the Yao-Yao graph $YY_6$, the Theta-Theta graph $\Theta\Theta_6$ defined by six cones is a spanner for sets of points in convex position. We also show that, for sets of points in non-convex position, the spanning ratio of $\Theta\Theta_6$ is unbounded. 
\end{abstract}

\section{Introduction}
Let $S$ be a set of $n$ points in the plane and let $G = (S, E)$ be a weighted geometric graph with vertex set $S$ and a set $E$ of (directed or undirected) edges between pairs of points, where the weight of an edge $uv \in E$ is equal to the Euclidean distance $|uv|$ between $u$ and $v$. 
The \emph{length} of a path in $G$ is the sum of the weights of its constituent edges. The distance $d_G(u, v)$ in $G$ between two points $u, v \in S$ is the length of a shortest path in $G$ between $u$ and $v$.
The graph $G$ is called a \emph{t-spanner} if any two points $u, v \in S$ at distance $|uv|$ in the plane are at distance $d_G(u, v) \le t\cdot |uv|$ in $G$. The smallest integer $t$ for which this property holds is called the \emph{spanning ratio} of 
$G$. 

The Yao graph $Y_k(S)$ and the Theta graph $\Theta_k(S)$ are defined for a fixed integer $k > 0$ as follows. 
Partition the plane into  $k$ equiangular cones 
by extending $k$ equally-separated rays starting at the origin, with the first ray in the direction of the positive $x$-axis. Then translate the cones to each point $u \in S$, and connect $u$ to a ``nearest'' neighbor in each cone. 
The difference between Yao and Theta graphs is in the way the ``nearest'' neighbor is defined. 
For a fixed point $u \in S$ and a cone $\cone(u)$ with apex $u$, a Yao edge $\overrightarrow{uv} \in \cone(u)$ minimizes the Euclidean distance $|uv|$ between $u$ and $v$, whereas a Theta edge $\overrightarrow{uv} \in \cone(u)$ minimizes the \emph{projective distance} $\|uv\|$ from $u$ to $v$, which is the Euclidean distance between $u$ and the orthogonal projection of $v$ on the bisector of $\cone(u)$. Ties are arbitrarily broken. 

Each of the graphs $\Theta_k$ and $Y_k$ has out-degree $k$, but in-degree $n-1$ in the worst case (consider, for example, the case of $n-1$ points uniformly distributed on the circumference of a circle centered at the $n^{th}$ point: for any $k \ge 6$, the center point has in-degree $n-1$). This is a significant drawback in certain wireless networking applications where a wireless node can communicate with only a limited number of neighbors. 
To reduce the in-degrees, a second filtering step can be applied to the set of incoming edges in each cone. This filtering step eliminates, for each each point $u \in S$ and each cone with apex $u$, all but a ``shortest'' incoming edge. The result of this filtering step applied on $\Theta_k$ ($Y_k$) is the Theta-Theta (Yao-Yao) graph $\Theta\Theta_k$ ($YY_k$). 
Again, the definition of ``shortest'' differs for Yao and Theta graphs: a shortest Yao edge $\overrightarrow{vu} \in \cone(u)$ minimizes $|vu|$, and a shortest Theta edge $\overrightarrow{vu} \in \cone(u)$ minimizes $\|vu\|$.  
Again, ties are arbitrarily broken. 

Yao and Theta graphs (and their Yao-Yao and Theta-Theta sparse variants) have many important applications
in wireless networking~\cite{AlzoubiMW03}, motion planning~\cite{Clark87} and walkthrough animations~\cite{FLZ98}.
We refer the readers to the books by Li~\cite{Li08} and Narasimhan and Smid~\cite{ns-gsn-07} for more details on their uses, and 
to the comprehensive survey by Eppstein~\cite{Epp00} for related topics on geometric spanners. 
Many such applications take advantage of the spanning and sparsity properties of these graphs, which have been extensively studied. 
Molla~\cite{MollaThesis09} showed that $Y_2$ and $Y_3$ may not be spanners, and her examples can be used to show that $\Theta_2$ and $\Theta_3$ are not spanners either. 
On the other hand, it has been shown that,  for any $k \ge 4$, $Y_k$ and $\Theta_k$ are spanners: 
$Y_4$ is a $54.6$-spanner~\cite{DamianNelavalli17} and $\Theta_4$ is a $17$-spanner~\cite{BC+18};
$Y_5$ is a $3.74$-spanner~\cite{Barba+15} and $\Theta_5$ is a $9.96$-spanner~\cite{BoseMRS14};
$Y_6$ is a $5.8$-spanner~\cite{Barba+15} and $\Theta_6$ is a $2$-spanner~\cite{BGH+10}; 
for $k \ge 7$, the spanning ratio of $Y_k$ is $\frac{1+\sqrt{2-2\cos(2\pi/k)}}{2\cos(2\pi/k)-1}$~\cite{BDDx10}
and the spanning ratio of  $\Theta_k$ is $\frac{1}{1-2\sin(\pi/k)}$~\cite{RS91}; 
improved bounds on the spanning ratio of $Y_k$ for odd $k \ge 5$, and for $\Theta_k$ for even $k \ge 6$, also exist~\cite{BoseRV13}. 

In contrast with Yao and Theta graphs, our knowledge of Yao-Yao and Theta-Theta graphs is more limited. 
Li et al.~\cite{li02sparse} proved that $YY_k$ is connected for $k > 6$ and provided substantial experimental
evidence suggesting that $YY_k$ is a spanner for large $k$ values. This conjecture has been partly confirmed 
by Bauer and Damian~\cite{DB13} who showed that, for $k \ge 6$, 
$YY_{6k}$ is a spanner with spanning ratio $11.76$. This spanning ratio has 
been improved to $7.82$ in~\cite{Damian18} for a more general class of graphs called \emph{canonical $k$-cone graphs},
which include both $YY_{6k}$ and $\Theta\Theta_{6k}$, for $k \ge 6$. The same paper establishes a spanning ratio of 
$16.76$ for $YY_{30}$ and $\Theta\Theta_{30}$. 
Recent breakthroughs show that $YY_{2k}$, for any $k \ge 42$, is a spanner with 
spanning ratio $6.03 + O(k^{-1})$~\cite{LiZhan16}, and $YY_k$ for odd $k \ge 3$ is not a spanner~\cite{Jin+18}. 
For small values $k \le 5$, Damian et al.~\cite{DMP09} show that $YY_4$ is not a spanner, and 
Barba et al.~\cite{Barba+15} show that $YY_5$ is not a spanner, and their constructions can also be used to show 
that $\Theta\Theta_4$ and $\Theta\Theta_5$ are not spanners.
Molla~\cite{MollaThesis09} showed that $YY_6$ is also not a spanner, even for sets of 
points in convex position. This paper fills in one of the gaps in our knowledge of 
Theta-Theta graphs and shows that $\Theta\Theta_6$ is an $8$-spanner for sets of points in convex position, but 
has unbounded spanning ratio for sets of points in non-convex position.  

\section{Definitions}
%
Throughout the paper, $S$ is a fixed set of $n$ points in the plane and $k > 1$ is a fixed integer.  
The graphs $Y_k$ and $\Theta_k$ use a set of $k$ equally-separated rays starting at the origin. These rays define $k$ equiangular cones $\cone_1, \cone_2, \ldots, \cone_k$, each of angle $\theta = 2\pi/k$, with the lower ray of $\cone_1$ extending in the direction of the positive $x$-axis. Refer to~\autoref{fig:t6-sectors}. 
We assume that each cone is half-open and half-closed, meaning that it includes the clockwise bounding ray, but it excludes the counterclockwise bounding ray. 
%
\begin{figure}[htbp]
\centering
\includegraphics[width=0.9\linewidth]{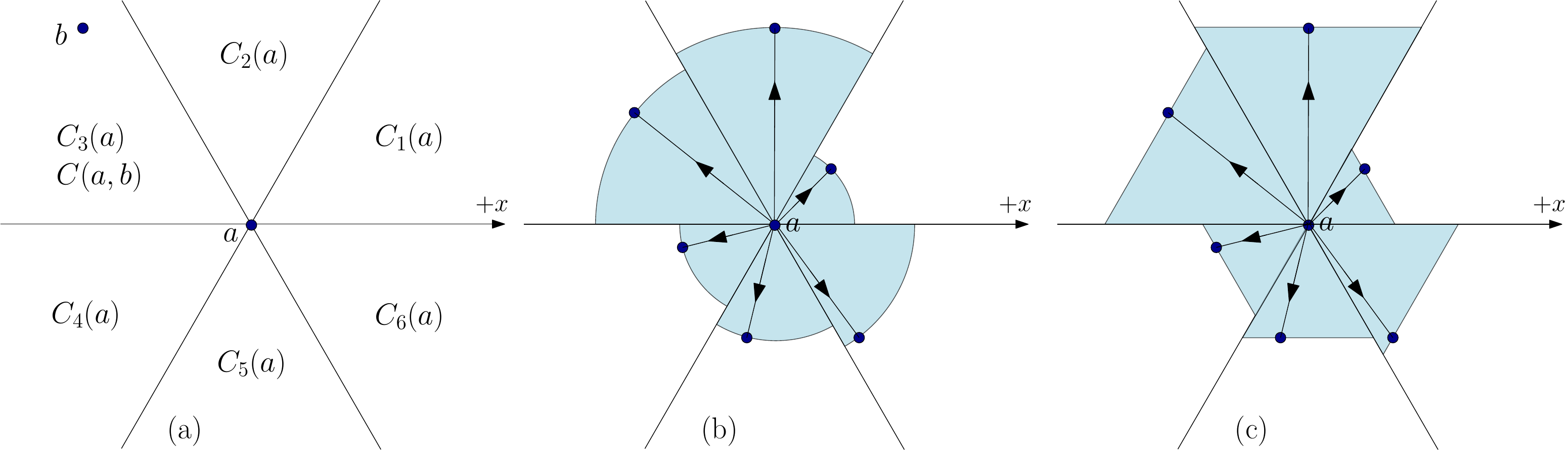} 
\caption{(a) Cones defining $Y_6$ and $\Theta_6$ (b) $Y_6$ edges minimize Euclidean distances (c) $\Theta_6$ edges minimize projective distances.}
\label{fig:t6-sectors}
\end{figure}
%
Let $\cone_i(a)$ denote a copy of $\cone_i$ translated to $a$, for each $a \in S$ and each $i = 1, \ldots, k$.
The directed graphs $\arr{Y_k}$ and $\arr{\Theta_k}$ are constructed as follows. 
In each cone $C_i(a)$, for each $i = 1, \ldots, k$ and each $a \in S$, extend a directed edge from $a$ to a ``nearest'' point $b$ that lies in $C_i(a)$. 
Yao and Theta graphs differ only in the way ``nearest'' is defined. 
A point $b$ is ``nearest'' to $a$ in $Y_k$ if it minimizes the Euclidean distance $|ab|$, whereas $b$ is ``nearest'' to $a$ in $\Theta_k$ if it minimizes the projective distance $\|ab\|$. See~\autoref{fig:graphs}a,b for simple graph examples illustrating these definitions.
\begin{figure}[htbp]
\centering
\includegraphics[width=0.5\linewidth]{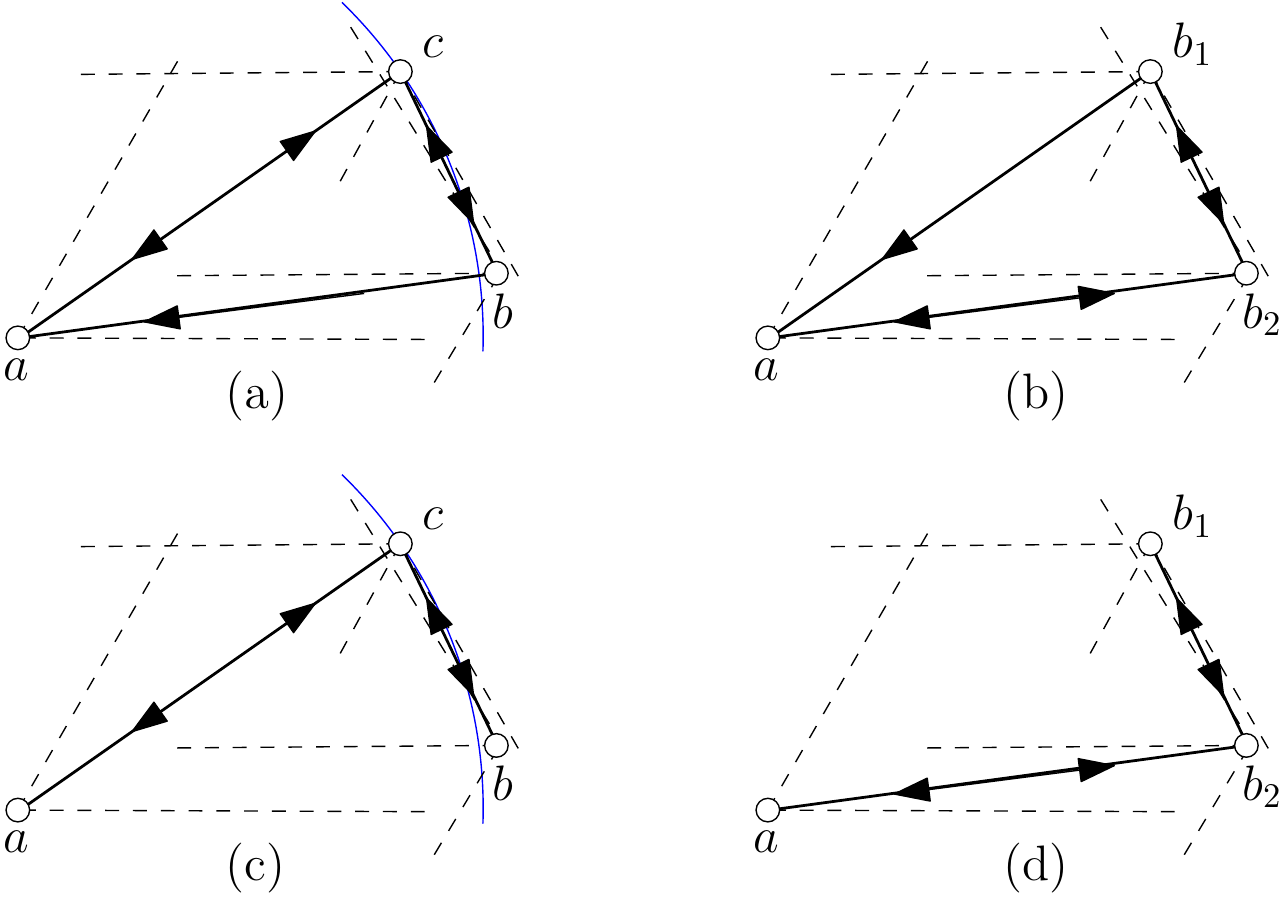} 
\caption{Graph examples (a) $Y_6$ (b) $\Theta_6$ (c) $YY_6$  (d) $\Theta\Theta_6$.}
\label{fig:graphs}
\end{figure}

The Yao-Yao graph $\arr{YY_k} \subseteq \arr{Y_k}$ and Theta-Theta graph $\arr{\Theta\Theta_k} \subseteq \arr{\Theta_k}$
are obtained by applying a filtering step to the set of incoming edges at each vertex in $\arr{Y_k}$ and $\arr{\Theta_k}$, respectively. 
Specifically, for each $a \in S$ and each $i = 1, \ldots, k$, these graphs retain a ``shortest'' incoming edge that lies in $\cone_i(a)$ and 
and discard the rest of incoming edges, if any. Recall that a ``shortest'' Yao edge $\overrightarrow{ba} \in \cone_i(a)$ minimizes $|ba|$, whereas a ``shortest'' Theta edge $\overrightarrow{ba} \in \cone_i(a)$ minimizes $\|ba\|$.~\autoref{fig:graphs}c(d) depicts the graph $YY_6$ ($\Theta\Theta_6$) after this filtering step has been applied to the graph $Y_6$ ($\Theta_6$) from~\autoref{fig:graphs}a(b).

\section{Background: $YY_6$ is not a Spanner}
\label{sec:yy6}
Molla~\cite{MollaThesis09} gave an example of a set of points in convex position for which $YY_6$ is not a spanner. We briefly review her construction here and show that the result does not hold for $\Theta\Theta_6$.
The construction 
begins with a strip of equilateral triangles between two horizontal lines with vertices $\{a_1, a_2, \ldots, a_n\}$ on the lower line (which we call the $a$-line) and $\{b_1, b_2, \ldots, b_n\}$ on the upper line (which we call the $b$-line). See the left of~\autoref{fig:yy6cex}a. Next the $a$-line is rotated clockwise about $a_1$ and the $b$-line is rotated counterclockwise about $b_1$ by a small angle $\alpha > 0$, to guarantee that 
$|a_{i-1}a_{i}| < |b_{i-1}a_{i}| $ and 
$|b_{i-1}b_{i}| < |a_ib_{i}| $, for $i = 2, \ldots, n$. 
The points are also slightly perturbed to ensure that $\cone_{2}(a_i)$ and $\cone_5(b_i)$ are all empty, for $i = 1, \ldots, n$. The result is depicted in the right of~\autoref{fig:yy6cex}a.

\begin{figure}[htbp]
\centering
\includegraphics[width=\linewidth]{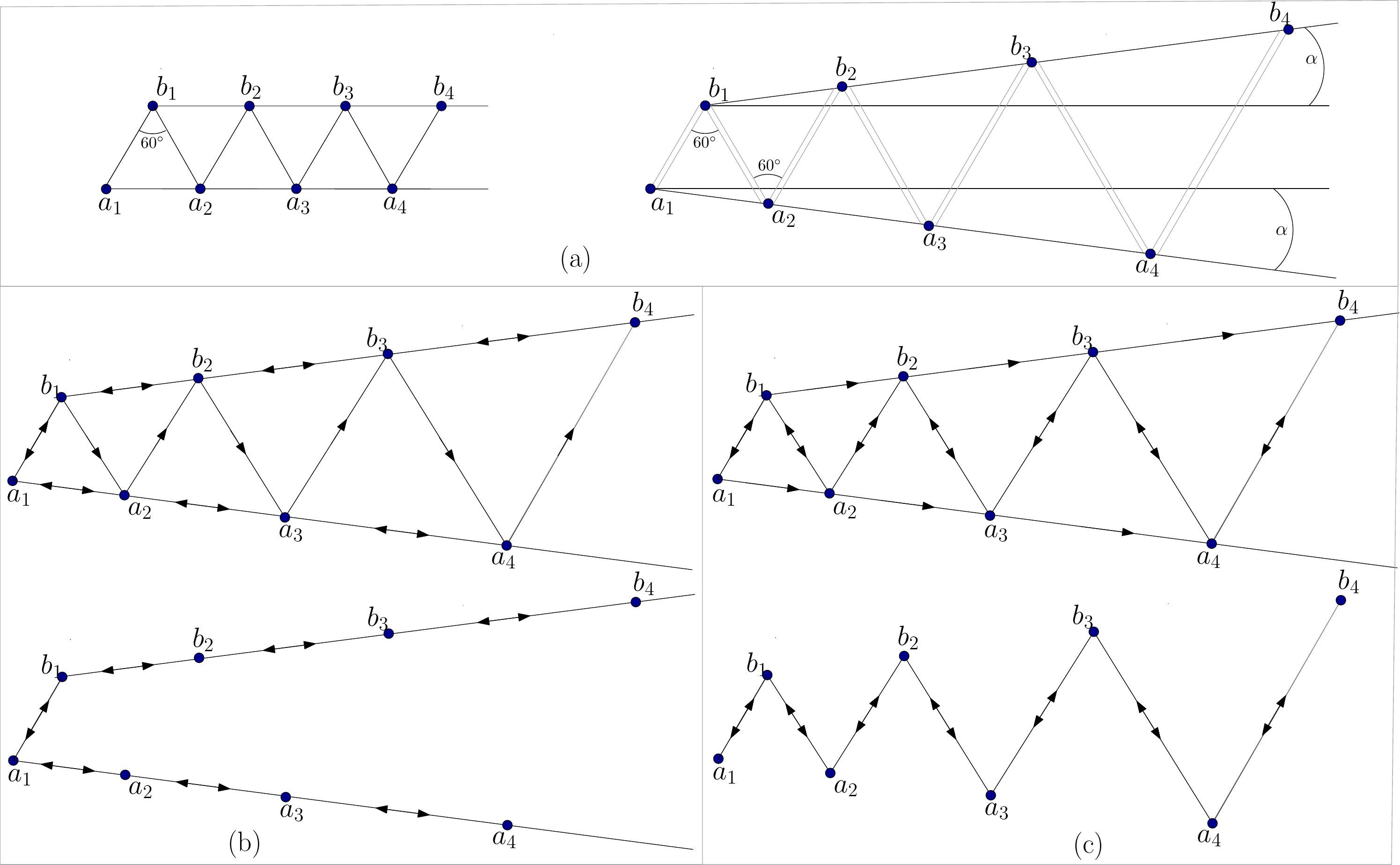} 
\caption{(a) Point set $\{a_1, \ldots a_4\} \cup \{b_1, \ldots, b_4\}$ (b) Graphs $Y_6$ (top) and $YY_6$ (bottom) 
(c) Graphs $\Theta_6$ (top) and $\Theta\Theta_6$ (bottom).
}
\label{fig:yy6cex}
\end{figure}

The graphs $Y_6$ and $YY_6$ induced by the set of points $S = \{a_1, \ldots a_4\} \cup \{b_1, \ldots, b_4\}$ are depicted in~\autoref{fig:yy6cex}b. Note that, with the exception of $a_1b_1$, $YY_6$ includes none of the $Y_6$ edges incident on both the $a$-line and the $b$-line. This is because, for $i > 1$, $\arr{b_{i-1}a_i}$ and $\arr{a_{i-1}a_i}$ both lie in $\cone_3(a_i)$ and $YY_6$ maintains only the shorter of the two, which is $\arr{a_{i-1}a_i}$.
Similarly, $\arr{a_ib_i}$ and $\arr{b_{i-1}b_i}$ both lie in $\cone_4(b_i)$ and $YY_6$ maintains only the shorter of the two, which is $\arr{b_{i-1}b_i}$. This shows that the shortest path in $YY_6$ between $a_n$ and $b_n$ is a Hamiltonian path of length at least $2n-1$, which grows arbitrarily large with $n$.  It follows that $YY_6$ is not a spanner. 

For the same point set $S$, the graphs $\Theta_6$ and $\Theta\Theta_6$ are depicted in~\autoref{fig:yy6cex}c. Note that, if projective distances are used, then 
$\|a_{i-1}a_{i}\| > \|b_{i-1}a_{i}\| $ and 
$\|b_{i-1}b_{i}\| > \|a_ib_{i}\| $, for $i = 2, \ldots, n$. 
These properties force $\Theta\Theta_6$ to maintain $\arr{b_{i-1}a_{i}} \in \cone_3(a_i)$ and $\arr{a_ib_{i}} \in \cone_4(b_i)$, for each $i = 2, \ldots, n$. The result is the zig-zag path depicted in~\autoref{fig:yy6cex}c which shows that, for this particular point set, $\Theta\Theta_6$ is a spanner. In the next section we show that $\Theta\Theta_6$ is a spanner for any set of points in convex position. 



\section{$\Theta\Theta_6$ is a Spanner for Points in Convex Position}
It has been established in~\cite{MollaThesis09} (and revisited in~\autoref{sec:yy6} of this paper) that $YY_6$ is not a spanner for sets of points in convex position. In this section we show that, unlike $YY_6$, the graph $\Theta\Theta_6$ is an $8$-spanner for sets of points in convex position (in the next section we will show that this result does not hold for sets of points in non-convex position). This is the first result that marks a difference in the spanning properties of $YY$-graphs and $\Theta\Theta$-graphs. 

Throughout this section, we assume that $S$ is a set of points in convex position. For simplicity, we also assume that the points in $S$ are in general position, meaning that no two points lie on a line parallel to one of the rays that define the cones. This implies that there is a unique nearest point in each cone of $\Theta_6$ and $\Theta\Theta_6$. 
We begin with a few definitions.

For any $a, b \in S$, let $\cone(a,b)$ denote the cone with apex $a$ that contains $b$. 
For any ordered pair of vertices $a$ and $b$, let $\T(a,b)$ be the \emph{canonical triangle} delimited by the rays bounding $\cone(a, b)$ and the perpendicular through $b$ on the bisector of $\cone(a, b)$. See~\autoref{fig:monpath}a. 
For a fixed point $a \in S$ and $i \in \{1, \ldots, k\}$, let $\pp_{\Theta_6}(a, i)$ denote the path in $\Theta_6$ that starts at $a$ and follows the $\Theta_6$-edges that lie in cones $C_i$. See, for example, the path 
$\pp_{\Theta_6}(a, 1)$ depicted in~\autoref{fig:monpath}. Note that this path is monotone with respect to the bisector of $C_i$. This along with the fact that the point set $S$ is finite implies that the path itself is finite and well defined. 
We say that two edges $ab$ and $cd$ \emph{cross} if they share a point other than an endpoint ($a$, $b$, $c$ or $d$).

\begin{figure}[htbp]
\centering
\includegraphics[width=0.7\linewidth]{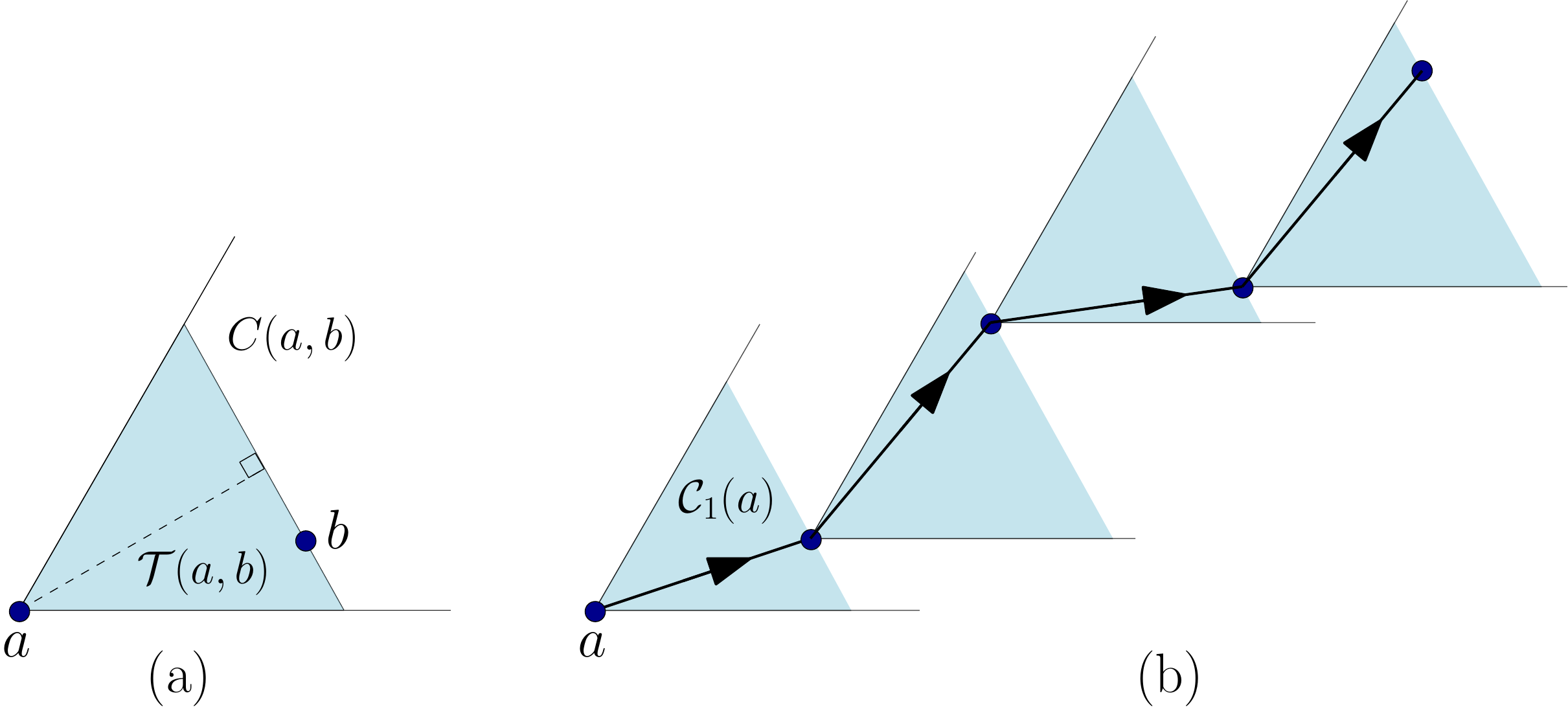} 
\caption{(a) Canonical triangle $\T(a,b)$ (b) Path $\pp_{\Theta_6}(a, 1)$.}
\label{fig:monpath}
\end{figure}

The \emph{half-$\Theta_6$-graph} introduced in~\cite{BGH+10} takes only ``half'' the edges of $\Theta_6$, those belonging to non-consecutive cones. Thus, the $\Theta_6$-graph is the union of two half-$\Theta_6$-graphs: 
one that includes all $\Theta_6$-edges that lie in cones $C_1$, $C_3$, $C_5$, and 
one that includes all $\Theta_6$-edges that lie in cones $C_2$, $C_4$, $C_6$. Bonichon et al.~\cite{BGH+10} show that half-$\Theta_6$ is a \emph{triangular-distance}\footnote{The \emph{triangular distance} from a point $a$ to a point $b$
is the side length of the smallest equilateral triangle centered at $a$ that touches $b$ and has one horizontal side.} 
Delaunay triangulation, computed as the dual of the Voronoi diagram based on
the triangular distance function. This, combined with Chew's proof that any triangular-distance Delaunay triangulation is a $2$-spanner~\cite{Chew89}, yields the following result.

\begin{theorem}{\emph{\cite{BGH+10}}}
The half-$\Theta_6$-graph is a plane $2$-spanner.
\label{thm:theta6}
\end{theorem}

\noindent
Next we introduce two preliminary lemma that will be useful in proving the main result of this section. 

\begin{figure}[htbp]
\centering
\includegraphics[width=0.99\linewidth]{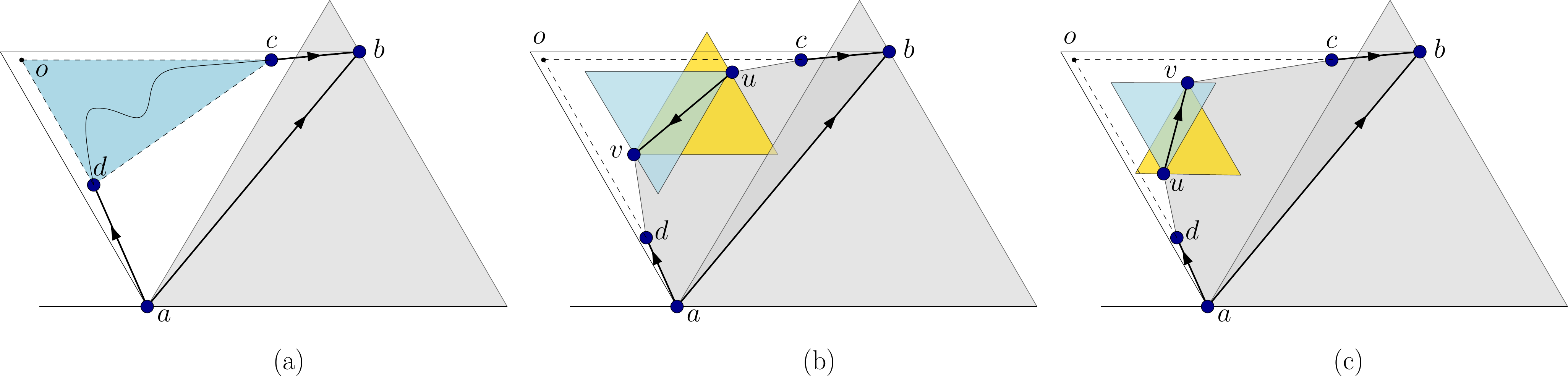} 
\caption{\autoref{lem:tripath} (a) There is a path $\pp_{\Theta\Theta_6}(d, c)$ that lies inside $\triangle{cod}$
(b) If $\arr{uv} \in \Theta_6$, $\arr{uv} \in \cone_4(u)$, then $\arr{uv} \in \Theta\Theta_6$ 
(c) If $\arr{uv} \in \Theta_6$, $\arr{uv} \in \cone_2(u)$, then $\arr{uv} \in \Theta\Theta_6$.}
\label{fig:tripath}
\end{figure}

\begin{lemma}
Let $S$ be a set of points in convex position and let $a, b, c, d \in S$ be distinct points such that $b \in \cone_1(a)$ and $\arr{ab} \in \Theta_6 \setminus \Theta\Theta_6$; $c \in \cone_4(b)$ and $\arr{cb} \in \Theta\Theta_6$; $d \in \cone_2(a)$ and $\arr{ad} \in \Theta_6$. Let $o$ be the intersection point between the upper ray of $\cone_4(c)$ and the left ray of $\cone_2(d)$. Then there is a path in $\Theta\Theta_6$ between $c$ to $d$ that lies in $\triangle{cod}$ and is no longer than $|oc| +|od|$. 
\label{lem:tripath}
\end{lemma}
\begin{proof}
Note that, since the points in $S$ are in convex position, the point $o$ exists and lies outside the convex quadrilateral $abcd$. Refer to~\autoref{fig:tripath}a. 
Consider the paths $\pp_c = \pp_{\Theta_6}(c, 4)$ and $\pp_d = \pp_{\Theta_6}(d, 2)$. Since $p_c$ and $p_d$ are in the same half-$\Theta_6$ graph, Theorem~\ref{thm:theta6} tells us that $\pp_c$ and $\pp_d$ do not cross. This implies that $\pp_c$ and $\pp_d$ meet in a point $e \in \triangle{cod}$. Let $\pp(c,e)$ be the piece of $\pp_c$ extending from $c$ to $e$, and $\pp(d,e)$ the piece of $\pp_d$ extending from $d$ to $e$. Note that $\pp(c,d) = \pp(c,e) \cup \pp(d,e)$ is a convex path that lies inside $\triangle{cod}$, which implies that $|\pp(c,d)| < |oc| + |od|$. 

To complete the proof, it remains to show that $p_{cd}$ is a path in $\Theta\Theta_6$. To do so, we consider an arbitrary edge $\arr{uv} \in \pp(c,d) \in \Theta_6$, and show that $\arr{uv} \in \Theta\Theta_6$. Assume first that $\arr{uv} \in \pp(c,e)$, meaning that $v \in \cone_4(u)$. Refer to~\autoref{fig:tripath}b. The convexity property of $S$ implies that no points may lie in $\T(v,u)$ and above $u$. 
Ignoring the piece of $\T(v,u)$ that extends above $u$, the rest of $\T(v,u)$ lies inside $\T(u,v) \cup abcuvd \cup T(a, b)$. This region, however, is empty of points in $S$: $\T(u,v)$ is empty of points in $S$ because $\arr{uv} \in \Theta_6$; $abcuvd$ is a convex polygon empty of points in $S$, by the convexity property of $S$; and $\T(a,b)$ is empty of points in $S$, 
because $\arr{ab} \in \Theta_6$. It follows that $\T(v,u)$ is empty of points in $S$ and therefore $\arr{uv} \in \Theta\Theta_6$.  

The arguments for the case when  $\arr{uv} \in \pp(d,e)$ are similar: in this case, $v \in \cone_2(u)$; no points in $S$ may lie in $\T(v,u)$ and left of $\cone_2(u)$; ignoring the piece of $\T(v,u)$ that extends left of $\cone_2(u)$, the rest of $\T(v,u)$ lies inside $\T(u,v) \cup abcvud \cup T(a, b)$, which is empty of points in $S$. It follows that $\T(v,u)$ is empty of points in $S$ and therefore $\arr{uv} \in \Theta\Theta_6$.  
\qed
\end{proof}

\begin{lemma}
For any edge $\arr{ab}$ in the $\Theta_6$-graph induced by a set of points $S$ in convex position, there is a path between $a$ and $b$ in $\Theta\Theta_6$ no longer than $4|ab|$. 
\label{lem:fullpath}
\end{lemma}
\begin{proof}
Assume without loss of generality that $\arr{ab} \in \cone_1(a)$ and let $\alpha$ be the angle formed by $ab$ with the lower ray of $\cone_1(a)$. 
Let $i_1$ ($h_1$) be the intersection point between the upper ray of $\cone_1(a)$ and the horizontal (perpendicular) through $b$. Refer to~\autoref{fig:fullpath}a. Let $i_2$ ($h_2$), $i_3$ ($h_3$) and $i_4$ ($h_4$) be copies of $i_1$ ($h_1$) rotated counterclockwise by $\pi/3$, $2\pi/3$ and $2\pi/3+\alpha$, respectively. Note that $|a h_1| < |ab|$ and $|bi_1| = 2 |i_1h_1|$. 
We show that there is a convex path $\pp(a,b) \in \Theta\Theta_6$ between $a$ and $b$ that lies inside the convex region $\R = abi_2i_3i_4$ (shaded in~\autoref{fig:fullpath}a). The length of such a path is 
\begin{eqnarray*}
|\pp(a,b)| & < & |bi_2| + |i_2i_3| + |i_3i_4| + |i_4a| \\
     	& = & 2|i_1h_1| + |i_1i_2| +  |i_2i_3| + |i_3i_4| + |i_4a| \\
	& < & 2|i_1h_1| + 4 |i_1i_2| < 4 (|i_1h_1| + |i_1i_2|) = 4|a h_1| \\
	& < & 4 |ab|
\end{eqnarray*}
\begin{figure}[htbp]
\centering
\includegraphics[width=0.95\linewidth]{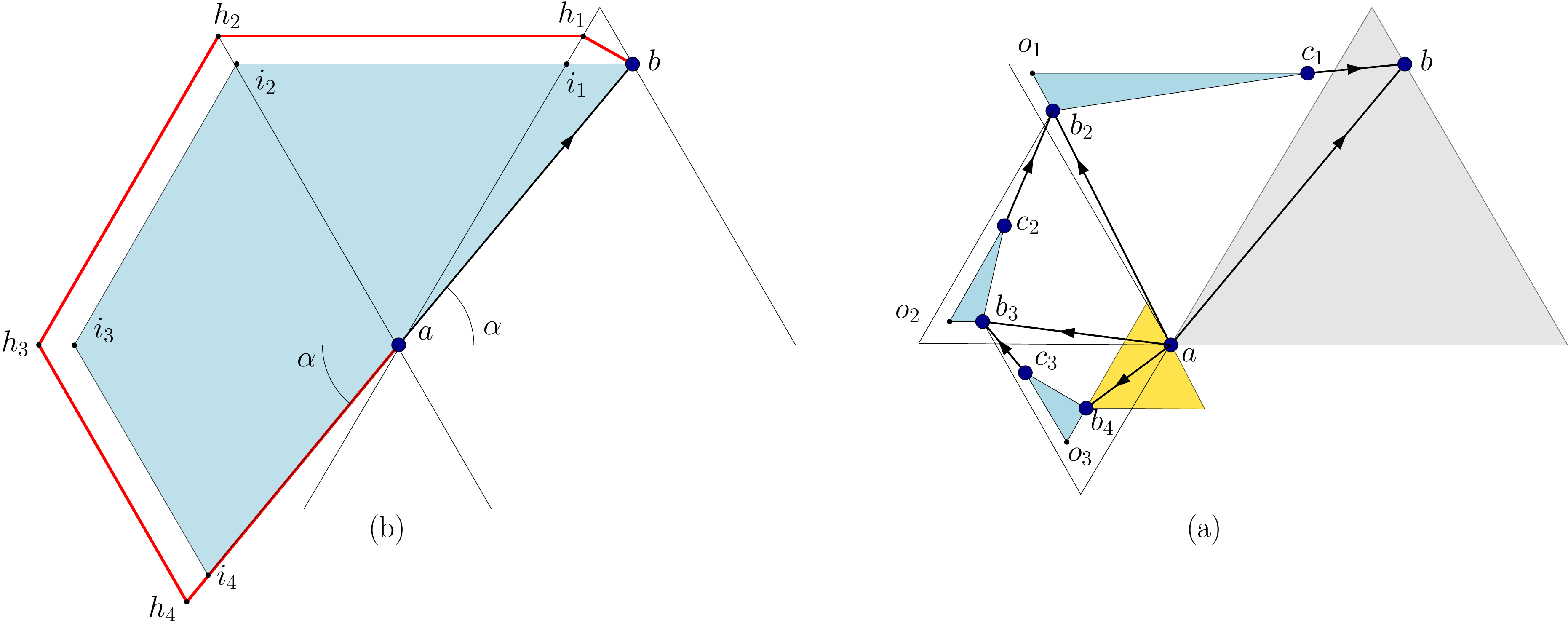} 
\caption{\autoref{lem:fullpath} (a) Any convex path from $a$ to $b$ that lies in the shaded area is no longer than $4|ab|$ (b) Path in $\Theta\Theta_6$ from $a$ to $b$: $\arr{ab} \in \Theta_6 \setminus \Theta\Theta_6$, $\arr{c_1b}, \arr{c_2b_2}, \arr{c_3b_3}, \arr{ab_4}  \in \Theta\Theta_6$.}  
\label{fig:fullpath} 
\end{figure}

\noindent
It remains to prove the existence of such a path $\pp(a,b) \in \Theta\Theta_6$. If $\arr{ab} \in \Theta\Theta_6$, then $\pp(a,b) = ab$ and the lemma trivially holds. Otherwise, there is  $\arr{c_1b} \in \Theta\Theta_6$, with $c_1 \in \cone(b,a)$. By definition $\|c_1b\| < \|ab\|$, which implies that $c_1$ lies in $\cone_2(a)$ or $\cone_6(a)$. Assume without loss of generality that $c_1 \in \cone_2(a)$; the case where $c_1 \in \cone_6(a)$ is symmetric. 
Because $\cone_2(a)$ is non-empty, $\Theta_6$ includes an edge $\arr{ab_2} \in \cone_2(a)$. Refer to~\autoref{fig:fullpath}b. If $b_2$ and $c_1$ coincide, let 
$\pp(b_2,c_1)$ be the empty path; otherwise, $\pp(b_2,c_1) \in \Theta\Theta_6$ is the path established by~\autoref{lem:tripath}, which lies in a triangular region inside $\T(a, c_1)$ (shaded in~\autoref{fig:fullpath}b). If $\arr{ab_2} \in \Theta\Theta_6$, then 
\[\pp(a,b) = ab_2 \oplus \pp(b_2,c_1) \oplus c_1b\]
is a convex path (by the convexity property of $S$) from $a$ to $b$ in $\Theta\Theta_6$ that lies inside $\R$, so the lemma holds. Here $\oplus$ denotes the concatenation operator. If $\arr{ab_2} \not \in \Theta\Theta_6$, then there is $\arr{c_2b_2} \in \Theta\Theta_6$, with $c_2 \in \cone(b_2,a)$. By definition $\|c_2b_2\| < \|ab_2\|$, which implies that $c_2 \in \cone_3(a)$.  
Because $\cone_3(a)$ is non-empty, $\Theta_6$ includes an edge $\arr{ab_3} \in \cone_3(a)$. If $b_3$ and $c_2$ coincide, let 
$\pp(b_3,c_2)$ be the empty path; otherwise, $\pp(b_3,c_2) \in \Theta\Theta_6$ is the path established by~\autoref{lem:tripath}, which lies inside $\T(a, c_2)$ (shaded in~\autoref{fig:fullpath}b). If $\arr{ab_3} \in \Theta\Theta_6$, then 
\[\pp(a,b) = ab_3 \oplus \pp(b_3,c_2) \oplus c_2b_2 \oplus \pp(b_2,c_1) \oplus c_1b\]
is a convex path from $a$ to $b$ in $\Theta\Theta_6$ that lies inside $\R$, so the lemma holds. If $\arr{ab_3} \not\in \Theta\Theta_6$, then there is $\arr{c_3b_3} \in \Theta\Theta_6$, with $c_3 \in \cone(b_3,a)$. By definition $\|c_3b_3\| < \|ab_3\|$, which implies that $c_3 \in \cone_4(a)$.  
Because $\cone_4(a)$ is non-empty, $\Theta_6$ includes an edge $\arr{ab_4} \in \cone_4(a)$. If $b_4$ and $c_3$ coincide, let 
$\pp(b_4,c_3)$ be the empty path; otherwise, $\pp(b_4,c_3) \in \Theta\Theta_6$ is the path established by~\autoref{lem:tripath}, which lies inside $\T(a, c_3)$. The convexity property of $S$ implies that the region of $\T(b_4, a)$ that extends right of the line supporting $ab$ is empty of points in $S$. Ignoring this region, the rest of $\T(b_4, a)$ lies in $\T(a, b_4) \cup \T(a, b_3)$, which is also empty of points in $S$. It follows that $\T(b_4, a)$ is empty of points in $S$, therefore $\arr{ab_4} \in \Theta\Theta_6$. These together imply that 
\[\pp(a,b) = ab_4 \oplus \pp(b_4,c_3) \oplus c_3b_3 \oplus \pp(b_3,c_2) \oplus c_2b_2 \oplus \pp(b_2,c_1) \oplus c_1b\]
 is a convex path from $a$ to $b$ in $\Theta\Theta_6$ that lies inside $\R$. This completes the proof.
\end{proof}
Lemmas~\ref{thm:theta6} and~\ref{lem:fullpath} together yield the main result of this section.
\begin{theorem}
The $\Theta\Theta_6$-graph induced by a set of points in convex position is an $8$-spanner. 
\label{thm:convexub}
\end{theorem}

The following lemma establishes a lower bound of $4$ on the spanning ratio of $\Theta\Theta_6$ for convex point sets. In addition, it shows that the bound $4$ of~\autoref{lem:fullpath} on the spanning ratio of $\Theta\Theta_6$-paths spanning $\Theta_6$-edges is tight. \begin{lemma}
The spanning ratio of the $\Theta\Theta_6$-graph induced by a set of points in convex position is at least $4$. 
\label{lem:convexlb}
\end{lemma}
\begin{proof}
We construct a set of points $S$ that satisfies the claim of this lemma. Let $a$ be an arbitrary point in the plane 
and let $b_i$ be the point at unit distance from $a$ that lies on the counterclockwise ray of $\cone_i(a)$, for $i = 1, \ldots, 4$. 
Refer to~\autoref{fig:tt6convextight}a. Perturb the points infinitesimally so that $b_1$ lies strictly inside $\cone_1(a)$ and $b_i$ lies strictly inside $\cone_i(a) \cap \T(a, b_{i-1})$, for $i = 2, 3, 4$. We ignore this infinitesimal quantity from our calculations and assume that $|ab_i| = 1$, for $i = 1, \ldots, 4$ and $|b_ib_{i+1}| = 1$, for $i = 1, 2, 3$. 

\begin{figure}[htbp]
\centering
\includegraphics[width=0.75\linewidth]{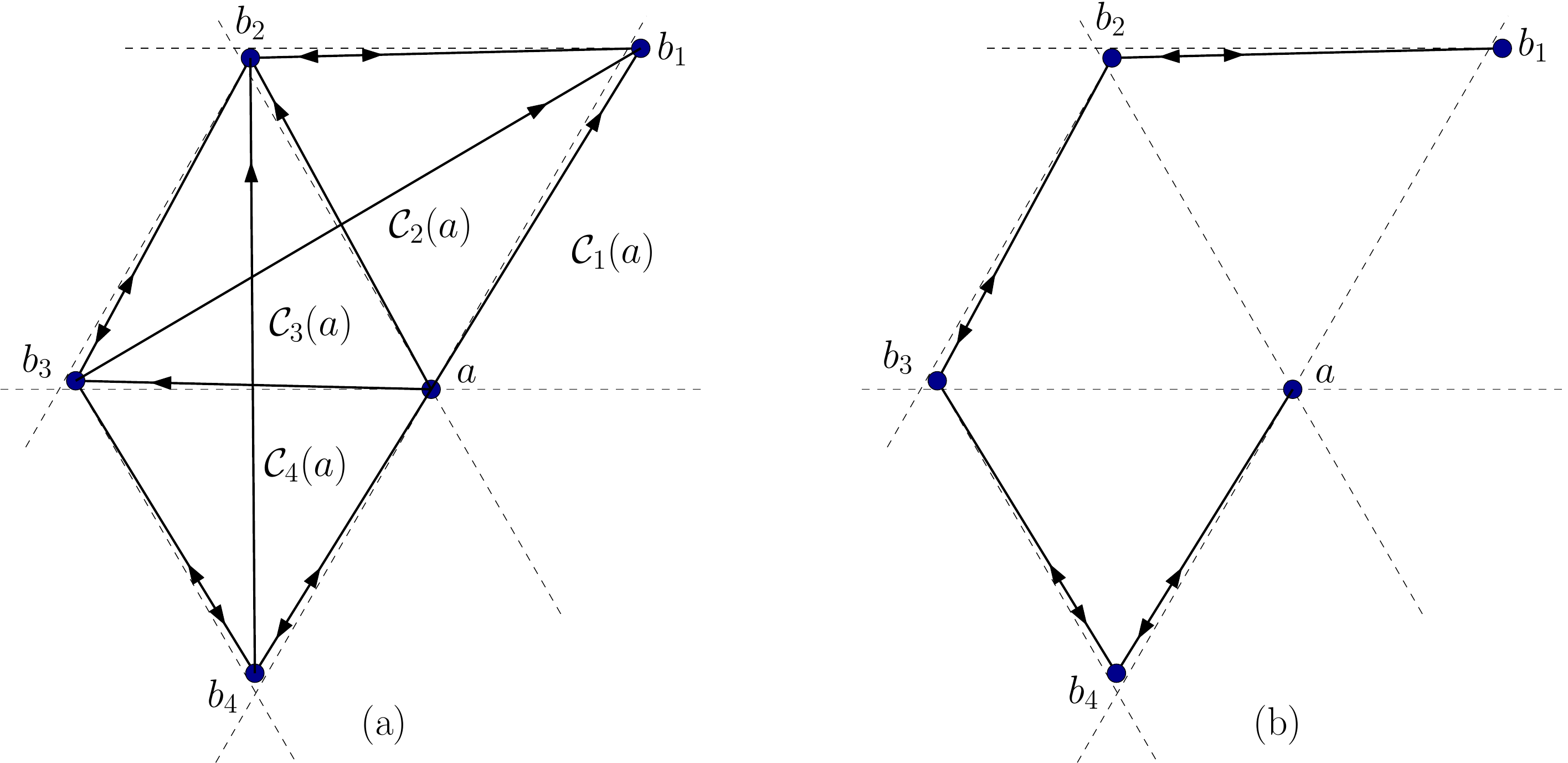} 
\caption{Set $S = \{a, b_1, b_2, b_3, b_4\}$ of points in convex position (a) $\Theta_6$-graph (b) $\Theta\Theta_6$-graph.}
\label{fig:tt6convextight} 
\end{figure}

Let $S = \{a, b_1, b_2, b_3, b_4\}$. The $\Theta_6$-graph and $\Theta\Theta_6$-graph induced by $S$ are depicted in~\autoref{fig:tt6convextight}a and~\autoref{fig:tt6convextight}b, respectively. Note that $\arr{ab_1} \in \Theta_6$, however 
$\arr{ab_1} \not\in \Theta\Theta_6$ and $\pp_{\Theta\Theta_6}(a,b_1) = ab_4 \oplus b_4b_3 \oplus b_3b_2 \oplus b_2b_1$ is a shortest path in $\Theta\Theta_6$ between $a$ and $b$ of length $4$. This proves the claim of this lemma. It also shows that the bound of~\autoref{lem:fullpath} is tight. 
\end{proof}

\section{$\Theta\Theta_6$ is not a Spanner for Points in Non-Convex Position}
In this section we show that there exist sets of points in non-convex position for which $\Theta\Theta_6$ has unbounded spanning ratio and therefore it is not a spanner. We show how to construct a set $S = \{a_i, b_i, c_i, d_i : i = 1, 2, \ldots, n\}$ of $4n$ points 
with this property. 

Let $a_1$ and $b_1$ be points in the plane such that $a_1b_1$ forms a $\pi/3$-angle with the horizontal through $a_1$. Let $r_a$ ($r_b$) be the ray with origin $a_1$ ($b_1$) pointing in the direction of the positive $x$-axis. Fix a small positive real value $0 < \alpha < 2$, and rotate $r_a$ ($r_b$) clockwise (counterclockwise) about $a_1$ ($b_1$) by angle $\alpha$. Let $a_2, a_3, \ldots a_{n}$ be points along $r_a$, and 
 $b_2, b_3, \ldots b_{n}$ points along $r_b$, such that $\angle{b_{i-1}a_ib_i} = \pi/3$ for each $i = 2, \ldots, n$, 
  and $\angle{a_ib_ia_{i+1}} = \pi/3$ for each $i = 1, 2, \ldots, n-1$. 
Refer to~\autoref{fig:tt6-init}a (which shows $\alpha$ enlarged for clarity). Note that at this point 
 $\cone_2(a_i)$ and $\cone_5(b_i)$ share the line segment $a_{i}b_{i}$, for each $i = 1, 2, \ldots, n$. Fix an arbitrary real value 
\begin{equation}
 \delta < \frac{|a_1a_2|\sin\alpha}{2}. 
 \label{eq:delta}
 \end{equation}
 Keep $a_1$ in place and shift the remaining points rightward alongside their supporting rays $r_a$ and $r_b$ such that 
 the horizontal distance between the right boundary ray of $\cone_2(a_i)$ and the left boundary ray of $\cone_5(b_i)$ is $\delta$, for each $i$. 
Refer to~\autoref{fig:tt6-init}b.
Finally, let $c_i$ ($d_i$) be a copy of $b_i$ ($a_i$) shifted upward (downward) by $2\delta$, 
for $i = 1, 2, \ldots, n$. Thus $b_ic_i$ and $a_id_i$ are vertical line segments of length $|b_ic_i| = |a_id_i| = 2\delta$. 
 
\begin{figure}[h]
\centering
\begin{tabular}{c@{\hspace{0.1\linewidth}}c}
\includegraphics[width=0.75\linewidth]{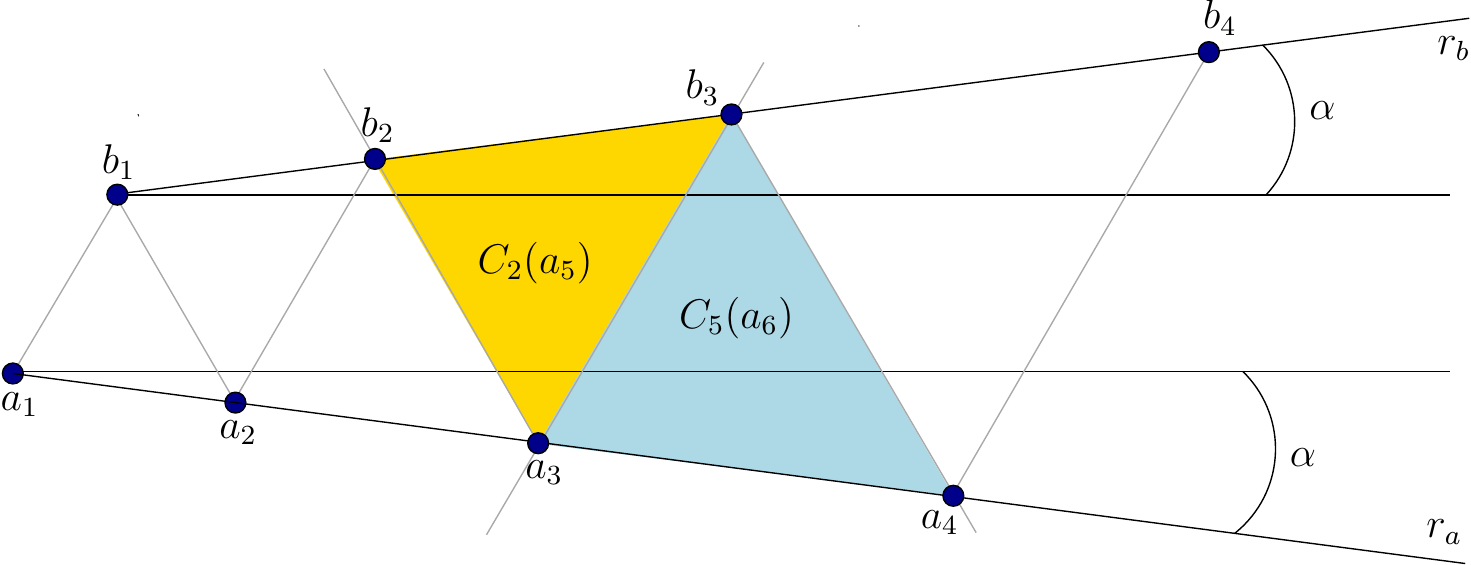} & \raisebox{5em}{(a)} \\
\includegraphics[width=0.75\linewidth]{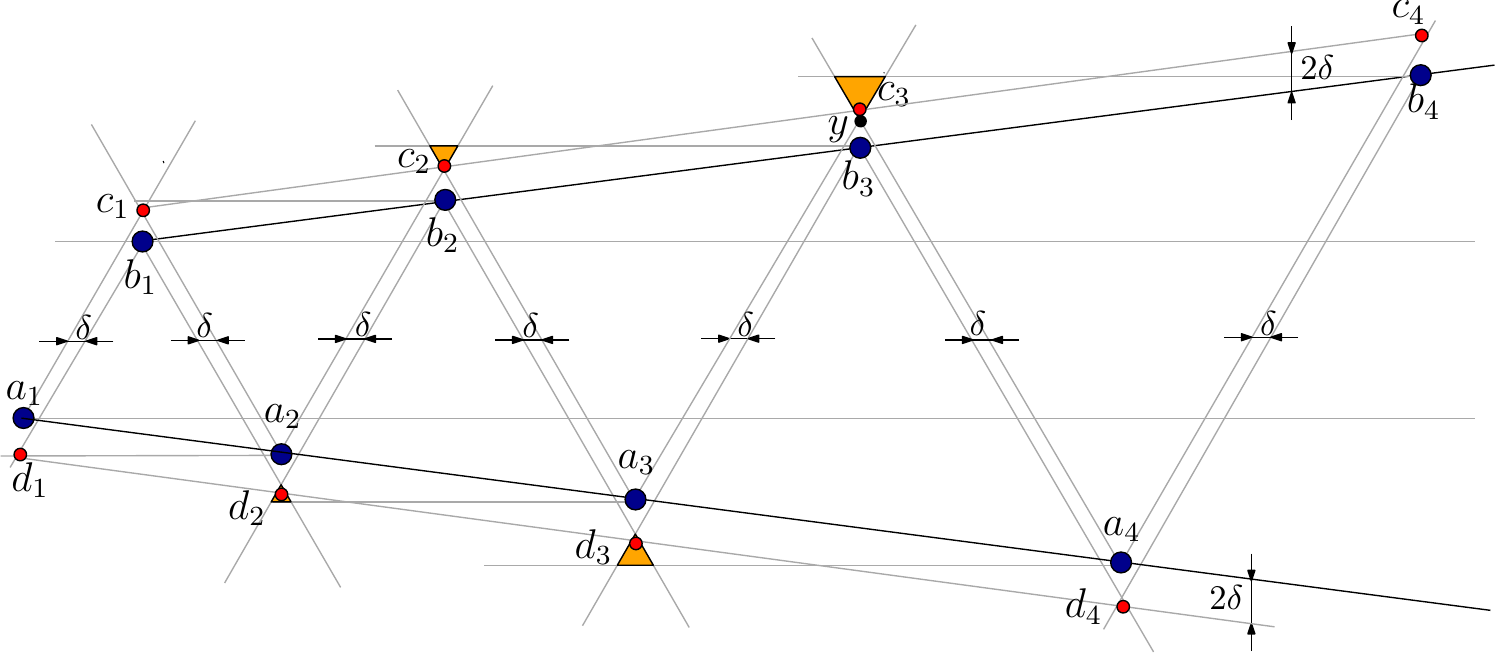} & \raisebox{6em}{(b)} 
\end{tabular}
\caption{(a) Initial point configuration (b) Shifted point positions.}
\label{fig:tt6-init}
\end{figure}

The following property is key to establishing an unbounded spanning ratio for $\Theta\Theta_6(S)$. 
\begin{property}
For each $i = 1, 2, \ldots, n-1$,  the point $c_i$ lies in a small triangular region at the intersection between 
$\cone_2(a_i)$, $\cone_2(a_{i+1})$ and $\cone_4(b_{i+1})$. 
\label{prop:c}
\end{property}
To establish this property, fix an arbitrary $i \in \{1, 2, \ldots, n-1\}$. 
Let $y$ be the intersection point between the right ray of $\cone_2(a_i)$ and the left ray of $\cone_2(a_{i+1}$). 
See~\autoref{fig:tt6-init}b, which depicts the instance $i = 3$. Note that $|b_iy|$ is equal to the height of an equilateral triangle of side length $2\delta$, which is $\delta\sqrt{3} < 2\delta = |b_ic_i|$. This means that $c_i$ lies vertically above $y$, therefore $c_i  \in \cone_2(a_i) \cap \cone_2(a_{i+1})$. 
To establish that $c_i \in \cone_4(b_{i+1})$, it suffices to show that $c_i$ lies below the horizontal line through $b_{i+1}$. The distance from $b_i$ 
to this line is $|b_{i}b_{i+1}|\sin\alpha > |a_1a_2|\sin\alpha > 2\delta$ (cf.~\autoref{eq:delta}). This along with the fact that $|b_ic_i| = 2\delta$ 
implies that $c_i$ lies below the horizontal line through $b_{i+1}$. This settles~\autoref{prop:c}. 

\medskip
Symmetric arguments establish the following property.
\begin{property}
For each $i = 2, \ldots, n-1$,  the point $d_i$ lies in a small triangular region at the intersection between 
$\cone_5(b_{i-1})$, $\cone_5(b_{i})$ and $\cone_3(a_{i+1})$. 
If $i = 1$, $d_i \in \cone_5(b_i) \cap \cone_3(a_{i+1})$. 
\label{prop:d}
\end{property}

\begin{figure}[htbp]
\centering
\begin{tabular}{c@{\hspace{0.1\linewidth}}c}
\includegraphics[width=0.75\linewidth]{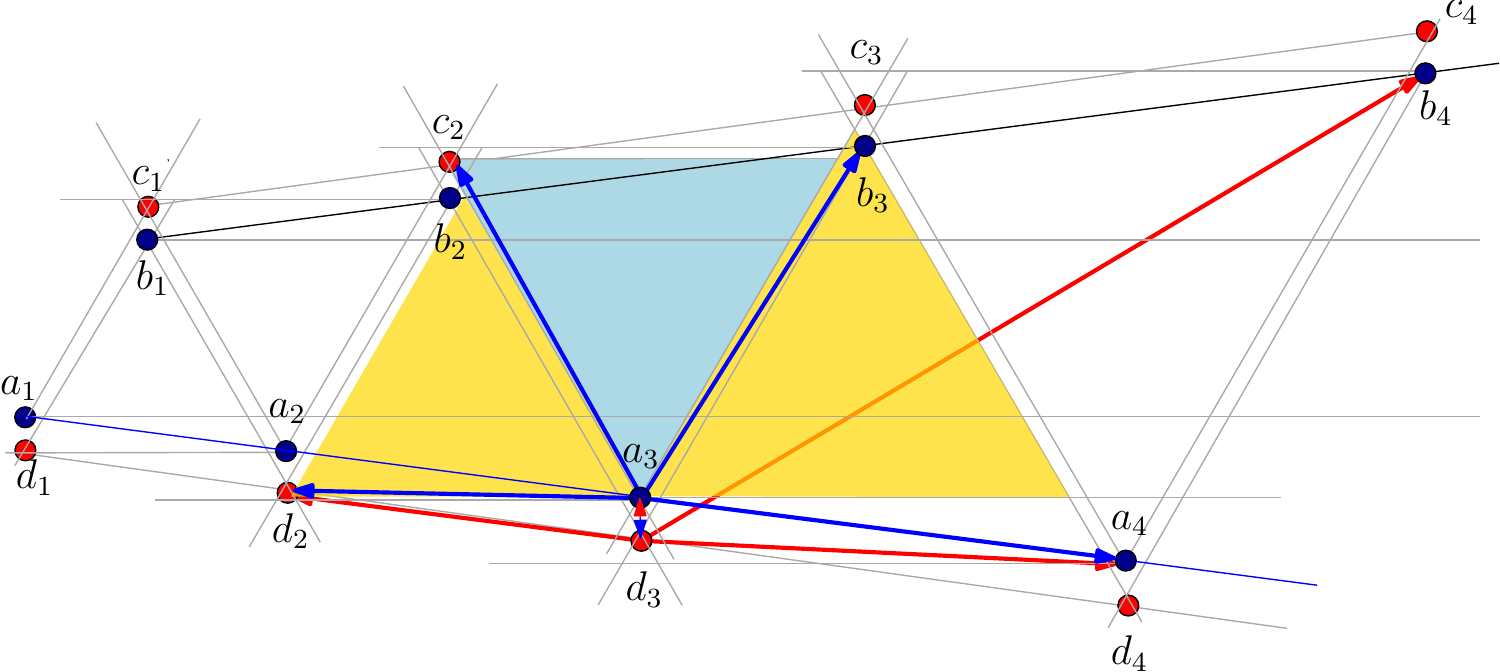} & \raisebox{6em}{(a)} \\
\includegraphics[width=0.75\linewidth]{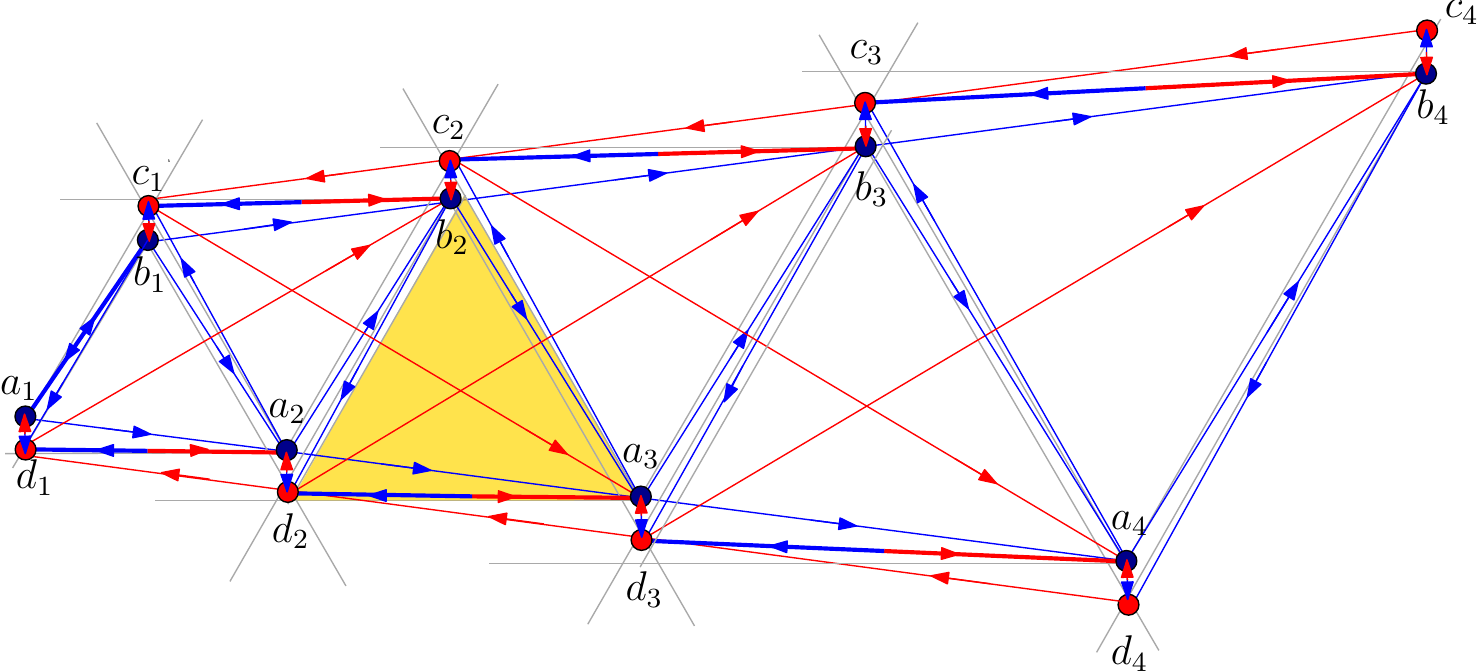} & \raisebox{6em}{(b)} \\
\includegraphics[width=0.75\linewidth]{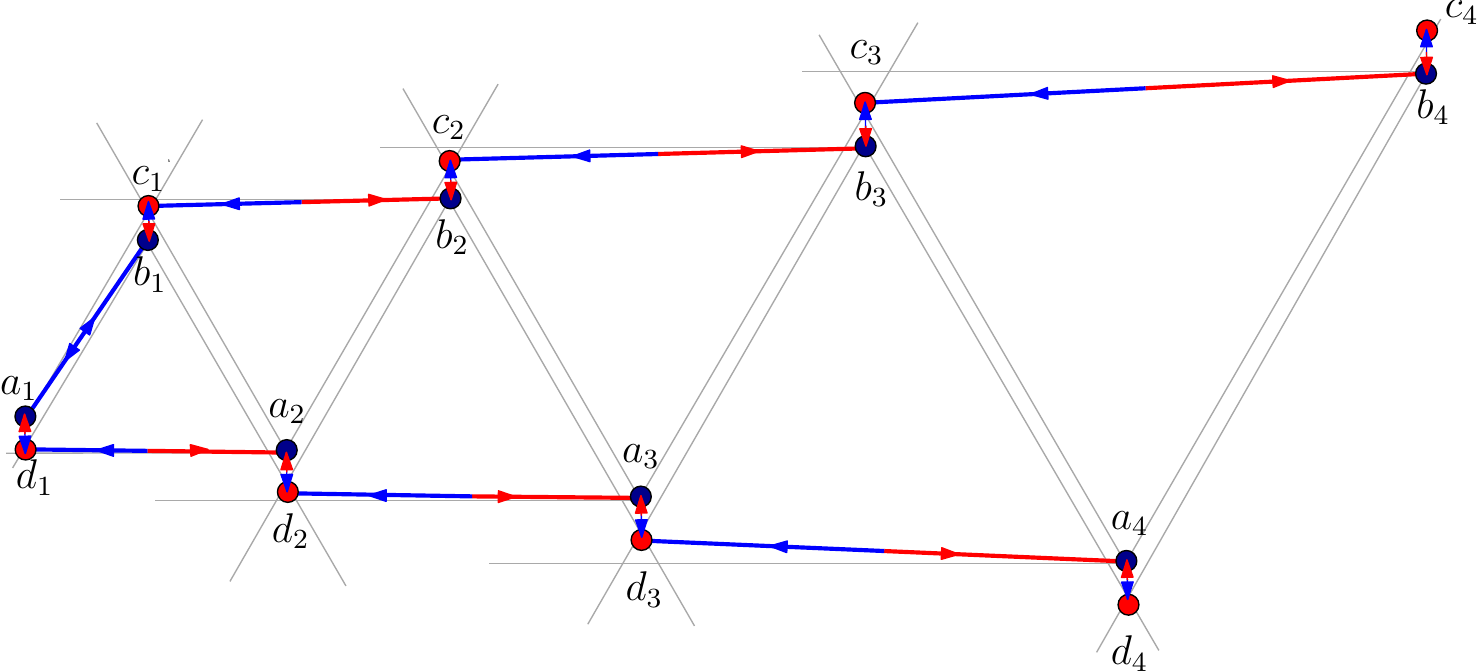} & \raisebox{6em}{(c)} \\
\end{tabular}
\caption{(a) Edges in $\Theta_6$ outgoing from $a_3$ and $d_3$ (b) $\Theta_6(S)$ (c) $\Theta\Theta_6(S)$.}
\label{fig:tt6-final} 
\end{figure}

We use Properties~\ref{prop:c} and~\ref{prop:d} in identifying the set of edges in $\Theta_6(S)$ and $\Theta\Theta_6(S)$. 
Fix an arbitrary $i \in \{1, \ldots, n\}$. The edges in $\Theta_6(S)$ outgoing from $a_i$ are: 
$\arr{a_ib_i} \in \cone_1(a_i)$; 
$\arr{a_ic_{i-1}} \in \cone_2(a_i)$, if $i > 1$; 
$\arr{a_id_{i-1}} \in \cone_3(a_i)$, if $i > 1$ 
(note that~\autoref{prop:d} implies that  $\|a_id_{i-1}\| < \|a_ib_{i-1}\|$);  
$\arr{a_id_i} \in \cone_5(a_i)$;  
and $\arr{a_ia_{i+1}} \in \cone_6(a_i)$, if $i < n$. Refer to~\autoref{fig:tt6-final}a, which 
depicts the instance $i = 3$. 
(Note that the cone $\cone_4(a_i)$ is empty of points in $S$.) 
The edges in $\Theta_6(S)$ outgoing from $d_i$ are: 
$\arr{d_ib_{i+1}} \in \cone_1(d_i)$, if $i < n$; 
$\arr{d_ia_i} \in \cone_2(d_i)$; 
$\arr{d_id_{i-1}} \in \cone_3(d_i)$, if $i > 1$;
and $\arr{d_ia_{i+1}} \in \cone_6(d_i)$, if $i < n$. 
(Note that the cones $\cone_4(d_i)$ and $\cone_5(d_i)$ are empty of points in $S$.)  
The edges in $\Theta_6(S)$ outgoing from $b_i$ are: 
$\arr{b_ib_{i+1}} \in \cone_1(b_i)$, if $i < n$; 
$\arr{b_ic_i} \in \cone_2(b_i)$; 
$\arr{b_ic_{i-1}} \in \cone_4(b_i)$, if $i > 1$ 
(note that~\autoref{prop:c} implies that  $\|b_ic_{i-1}\| < \|b_ia_i\|$);  
$\arr{b_id_{i}} \in \cone_5(b_i)$;
and $\arr{b_ia_{i+1}} \in \cone_6(b_i)$, if $i < n$. 
(Note that the cone $\cone_3(b_i)$ is empty of points in $S$.)
Finally, the edges in $\Theta_6(S)$ outgoing from $c_i$ are: 
$\arr{c_ib_{i+1}} \in \cone_1(c_i)$, if $i < n$; 
$\arr{c_ic_{i-1}} \in \cone_4(c_i)$, if $i > 1$; 
$\arr{c_ib_{i}} \in \cone_5(c_i)$;
and $\arr{c_ia_{i+2}} \in \cone_6(c_i)$, if $i < n-1$. 
(Note that the cones $\cone_2(c_i)$ and $\cone_3(c_i)$ are empty of points in $S$.)~\autoref{fig:tt6-final}b depicts the graph $\Theta_6(S)$, for $n = 4$. 

We now turn our attention to the set of incoming edges at each vertex in $\Theta_6(S)$. 
From among the four (three) edges directed into $a_i$ and lying in $\cone_3(a_i)$ for $i > 2$ 
($i = 2$), the edge $\arr{d_{i-1}a_i}$ has the shortest projective distance: 
$\|\arr{d_{i-1}a_i}\| < \|\arr{b_{i-1}a_i}\| < \|\arr{a_{i-1}a_i}\|$, and this latter quantity is 
in turn smaller than $\|\arr{c_{i-2}a_i}\|$, for $i > 2$. This implies 
that $\Theta\Theta_6(S)$ keeps $\arr{d_{i-1}a_i}$ and eliminates the other three (two) edges, for 
$i > 2$ ($i = 2$). Note that any cone with apex $a_i$ other than $\cone_3(a_i)$ contains at most one 
edge directed into $a_i$, which continues to exist in $\Theta\Theta_6$. 

For each $i$, the two edges directed into $d_i$ that lie in $\cone_2(d_i)$ satisfy $\|\arr{a_{i}d_i}\| < \|\arr{b_{i}d_i}\|$, 
therefore $\arr{b_{i}d_i}$ gets eliminated from $\Theta\Theta_6(S)$ in favor of $\arr{a_{i}d_i}$. Similarly, 
for $i < n$, $\arr{d_{i+1}d_i} \in \cone_6(d_i)$ gets eliminated from $\Theta\Theta_6(S)$ in favor of $\arr{a_{i+1}d_i} \in \cone_6(d_i)$. 
There are no edges in $\Theta_6(S)$ directed into $d_i$ that lie in any of the cones $\cone_1(d_i)$, 
$\cone_3(d_i)$, $\cone_4(d_i)$ and $\cone_5(d_i)$. 

For $i > 1$, the four edges directed into $b_i$ that lie in $\cone_4(b_i)$ satisfy 
$\|\arr{c_{i-1}b_i}\| < \|\arr{a_{i}b_i}\| < \|\arr{b_{i-1}b_i}\| < \|\arr{d_{i-1}b_i}\|$. 
This implies that $\Theta\Theta_6(S)$ keeps $\arr{c_{i-1}b_i}$ and eliminates the other three 
edges. The only other edge directed into $b_i$ is $\arr{c_ib_i} \in \cone_2(b_i)$. 
For $i = 1$, the two edges directed into $b_i$ are $\arr{a_ib_i} \in \cone_4(b_i)$ and 
$\arr{c_ib_i} \in \cone_2(b_i)$, which remain in place in $\Theta\Theta_6(S)$. 
Finally, $\Theta\Theta_6(S)$ eliminates $\arr{a_{i+1}c_i} \in \cone_5(c_i)$ in favor of $\arr{b_ic_i} \in \cone_5(c_i)$, and
$\arr{c_{i+1}c_i} \in \cone_1(c_i)$ in favor of $\arr{b_{i+1}c_i} \in \cone_1(c_i)$, for $i < n$. For $i = n$, the only 
edge directed into $c_i$ is $\arr{b_ic_i}$. 

The resulting $\Theta\Theta_6$-graph is the path depicted in~\autoref{fig:tt6-final}c. 
The edge set of  $\Theta\Theta_6(S)$ is $\{a_1b_1\} \cup \{a_id_i, b_ic_i : i = 1, 2, \ldots, n\} \cup \{d_ia_{i+1}, c_ib_{i+1}: i = 1, 2, \ldots, n-1\}$. 

For an arbitrarily small $\alpha$ value, we have $|a_nb_n| \approx |a_1b_1|$. 
A shortest path $\spp_{\Theta\Theta_6}(a_n,b_n)$ in this graph between $a_n$ and $b_n$ has length 
\begin{eqnarray*}
|\spp_{\Theta\Theta_6}(a_n,b_n)| & > & |a_1b_1| + \sum_{i=1}^{n-1}{\left(|a_id_i|+|d_ia_{i+1}|\right)} + 
\sum_{i=1}^{n-1}{\left(|b_ic_i|+|c_ib_{i+1}|\right)}  \\
& > & |a_1b_1| +  \sum_{i=1}^{n-1}{|a_ia_{i+1}|} +  \sum_{i=1}^{n-1}{|b_ib_{i+1}|} ~\mbox{~~~(by triangle inequality)}\\ 
& > & (2n-1)\cdot |a_1b_1| 
\end{eqnarray*}
This shows that the spanning ratio of $\Theta\Theta_6(S)$ is $\Omega(n)$, therefore we have the following result.

\begin{theorem}
The $\Theta\Theta_6$-graph is not a spanner. 
\end{theorem}

\section{Conclusions}
This paper establishes the first result showing a difference in the spanning properties of two related classes of sparse graphs, namely Yao-Yao and Theta-Theta. Previous results show $YY_k$ and $\Theta\Theta_k$ are not spanners for $k \le 5$, and are spanners for some values of $k > 6$. In this paper we show that, unlike $YY_6$, the graph $\Theta\Theta_6$ is a spanner for sets of points in convex position. We also show that, for sets of points in non-convex position, $\Theta\Theta_6$ is not a spanner. The spanning ratios of $YY_k$ and $\Theta\Theta_k$, for all even $k$ in the range $[8, 28]$ and for some even values of $k$ (those that are not multiples of $6$) in the range $[32, 82]$, remain unknown. 

\paragraph{Acknowledgement.} 
This work was initiated at the \emph{Third Workshop on Geometry and Graphs}, held at the Bellairs
Research Institute, March 8-13, 2015. We are grateful to the other workshop participants
for providing a stimulating research environment.

\section{Bibliography}
\bibliographystyle{plain}
\bibliography{spannerbib}

\end{document}